\documentclass[11pt]{article}
\usepackage{amsmath,amsfonts,amsthm,amssymb}
\usepackage{graphicx}
\usepackage[margin=2.4cm]{geometry}
\usepackage[bf, small]{caption}

\usepackage{multirow}
\usepackage{algorithm, algpseudocode, authblk}
\usepackage{url}
\usepackage{appendix}
\usepackage{hyperref}
\usepackage{float}

\usepackage{natbib}

\usepackage{setspace}
\usepackage{todonotes}

\newcommand{\beq}{\begin{equation}}
\newcommand{\eeq}{\end{equation}}
\newcommand{\beqn}{\begin{equation*}}
\newcommand{\eeqn}{\end{equation*}}
\renewcommand{\cite}[1]{\citet{#1}}
\newcommand{\refer}[1]{(\ref{#1})}
\newcommand{\R}{\mathbb{R}}

\newcommand{\ssection}[1]{\section[#1]{\centering #1}}

\newcommand{\mode}{\mbox{$M$} }

\newtheorem{theorem}{Theorem}
\newtheorem{example}{Example}
\newtheorem{assumption}{Assumption}
\newtheorem{lemma}[theorem]{Lemma}

\newtheorem{remark}[theorem]{Remark}
\newcommand{\D}{\mathcal{D}}
\newcommand{\change}[1]{{\color{black} #1}}
\newcommand{\dsum}{\displaystyle \sum}

\DeclareMathOperator{\Unif}{\operatorname{Uniform}}

\title{\Large Computing semiparametric bounds on the expected payments of insurance instruments via column generation\footnote{This work has been carried out thanks to funding from the Casualty Actuarial Society, Actuarial Foundation
of Canada, and Society of Actuaries}}

\author[1]{Robert Howley\thanks{howley.robert@gmail.com}}
\author[1]{Robert Storer\thanks{rhs2@lehigh.edu}}
\author[2]{Juan Vera\thanks{j.c.veralizcano@uvt.nl}}
\author[1]{Luis F. Zuluaga\thanks{luis.zuluaga@lehigh.edu, Corresponding author.}}
\affil[1]{Department of Industrial and Systems Engineering, Lehigh University \authorcr H.S. Mohler Laboratory, 200 West Packer Avenue, Bethlehem, PA 18015}
\affil[2]{Department of Econometrics and Operations Research, Tilburg University \authorcr 5000 LE TILBURG The Netherlands}

\date{}

\begin{document}
\doublespacing
\maketitle

\begin{abstract}
\begin{singlespace}
\noindent It has been recently shown that numerical semiparametric bounds on the expected 
payoff of financial or actuarial instruments
can be computed using semidefinite programming. However, this approach has practical limitations. 
Here we use column generation, a classical optimization technique, to 
address these limitations.
From column generation, it follows that  practical univariate semiparametric bounds can be found
by solving a series of linear programs. In addition to moment information, the column generation approach allows the inclusion of extra information about the random variable; for instance, unimodality and continuity, as well as the construction of corresponding worst/best-case distributions in a simple way.
\end{singlespace}
\end{abstract}

\ssection{Introduction}
Many financial and insurance instruments protect against underlying losses for which it is difficult to make exact distributional assumptions. Under these circumstances, it is difficult to provide a good estimate of the loss distribution, which in turn makes it difficult to estimate payments on the corresponding insured loss. Computing {\em semiparametric bounds} on the expected payments is an approach that has been successfully used to deal with this problem. This involves finding the minimum and maximum expected payments on the insurance instrument, when only partial information (e.g., moments) of the underlying loss distribution is known. For example, consider the work of \cite{cox, jansen,vill}. This approach has also been used to address the estimation of bounds on extreme loss probabilities (\cite{coxzul10}), and the prices of insurance instruments, and financial options (\cite{brock,lo,schepper}). These semiparametric bounds are useful when the structure of the product is too complex to develop analytical or simulation based valuation methods, or when it is difficult to make strong distributional assumptions on the underlying risk factors. Furthermore, even when distributional assumptions can be made, and analytical valuation formulas or simulation based prices can be derived, these bounds are useful to check the consistency of such assumptions.

The semiparametric bound approach is also referred as {\em distributionaly-robust} \citep[see, e.g.,][]{Ye10} or {\em ambiguity-averse} \citep[see, e.g.,][]{Nata11}. Also, it has been shown that this approach partially reflects the manner in which persons naturally make decisions \citep[cf.,][]{Nata11}

In the actuarial science and financial literature, there are two main approaches used to compute semiparametric bounds: analytically, by deriving closed-form formulas for special instances of the problem \citep[see, e.g.,][]{zhang11,cox,schepper}; and numerically, by using {\em semidedefinite programming} techniques \citep[cf.,][]{todd01} to solve general instances of the problem \citep[see, e.g.,][]{bert02, BoylL97,coxzul11,coxzul10}. 
An alternative numerical approach to solve semiparametric bounds proposed in the {\em stochastic
programming} literature by \citet{BirgD91}, based on the classical {\em column generation} (CG)
approach for mathematical optimization problems (see, e.g., \citet[Chp. 22]{Dantzig}, \citet{colgen}), has received little attention in the financial and
actuarial science literature.

Here, we 
consider the use of CG to obtain semiparametric bounds
in the context of financial and actuarial science applications.
In particular, we show that for all practical purposes, univariate semiparametric bounds can be found
by solving a sequence of linear programs associated to the CG {\em master} problem (cf., Section~\ref{sec:method}).
We also show that the CG approach allows the inclusion of extra information about the random variable such as unimodality and continuity, as well as the construction of the corresponding worst/best-case distributions in a simple way.
Also, the CG methodology achieves accurate results at a very small computational cost, it is straightforward to implement,
and the core of its implementation remains the same for very general, and practical instances of semiparametric bound
problems. 

To illustrate the potential of the CG approach,  
in Section~\ref{sec:ex1m2}, semiparametric lower and upper bounds are computed for the loss elimination ratio of a right censored deductible insurance policy, when the underlying risk distribution
is assumed to be unimodal, and have known first and second-order moments.
Furthermore, in Section~\ref{sec:worstcase}, we illustrate how continuous representations of
the worst/best-case distributions associated with the semiparametric bounds can be readily constructed and analyzed.

\ssection{Problem Description}
Consider a random variable $X$ with an unknown underlying distribution $\pi$,
but known support $\D \subseteq \R$ (not necessarily finite), and interval estimates $[\sigma_j^-,\sigma_j^+]$, $j = 1, \ldots, m$ for the expected value
of functions $g_j:\mathbb{R} \rightarrow \mathbb{R}$ for $j=1,\dots,m$ (e.g., typically, $g_j(x) = x^j$). 
The {\em upper semiparametric bound} on the expected value of the (target) function $f:\mathbb{R} \rightarrow \mathbb{R}$ is defined as:
\beq \label{eq:ub}
	\begin{aligned}
		B^*:=~ & \underset{\pi}{\text{sup}} & & E_\pi \left(f(X)\right) \\
		& \text{s.t.} & & \sigma_j^- \leq E_\pi \left(g_j(X)\right) \leq \sigma_j^+ \quad \quad j = 1, \ldots, m \\
		& & & \pi \text{ a probability distribution on } \D,
	\end{aligned}
\eeq
where $E_\pi (\cdot)$ represents the expected value under the distribution $\pi$.
That is, the upper semiparametric bound of the function $f$ is calculated by finding the supremum of $E_\pi\left(f(X)\right)$ across all possible probability distributions $\pi$, with support \change{on the set} $\D$, that satisfy the $2m$ expected value constrains. 
The parameters $\sigma_j^-,\sigma_j^+$, $j=1,\dots,m$ allow for 
confidence interval estimates for the expected value of $E_\pi(g_j(X))$,
that are not typically considered in the 
analytical solution of special instances of \refer{eq:ub} 
(i.e., typically $\sigma_j^- = \sigma_j^+$ in analytical solutions). 

The {\em lower semiparametric bound} of the function $f$ is formulated as the corresponding minimization problem, that is, by changing the $\sup$ to $\inf$ in the objective of \eqref{eq:ub}. 
We will provide details about the solution of the upper semiparametric bound problem \refer{eq:ub} that apply in analogous fashion to the corresponding lower semiparametric bound problem.
Also, for ease of presentation we will at times refer to both the upper and lower bound semiparametric problems using \refer{eq:ub}.

While specific instances of \eqref{eq:ub} have been solved analytically \change{\citep[see, e.g.,][]{lo, cox, schepper, zhang11}}, semidefinite programming (SDP) is currently the main approach
used in the related literature to numerically solve the general problem being considered here
\change{\citep[c.f.,][]{BoylL97,bert02,popescu}} whenever the functions $f(\cdot)$ and $g_j(\cdot)$ are 
piecewise polynomials. 
However, the SDP approach has important limitations in terms of the capacity of practitioners to use it. First, there are no commercially available SDP solvers. Second, the formulation of the SDP that needs to be solved for a given problem is not ``simple'' (\cite{coxzul10}) and must be re-derived for different support sets $\D$ of the distribution of $X$ \citep[see, e.g.,][Proposition 1]{bert02}.

\citet{BirgD91} proposed an alternative numerical method to solve the semiparametric bound
problem \refer{eq:ub} by using a CG
approach (see, e.g., \citet[Chp. 22]{Dantzig}, \citet{colgen}) that has received little attention in the financial and
actuarial science literature. Here, we show that the CG solution approach
addresses the limitations of the SDP solution approach discussed above.
Additional 
advantages of the CG solution approach in contrast to SDP techniques will be discussed at the end of Section~\ref{sec:method}.

\change{It is worth to mention that although in the next section we present the proposed algorithm in pseudo-algorithmic form (e.g., see Algorithm~\ref{fig:alg}), our implementation of the algorithm is available upon request to the authors.}

\ssection{Solution via column generation} \label{sec:method}
In this section we present the CG solution approach proposed
by \citet[Sec. 3]{BirgD91} to solve the semiparametric bound problem \refer{eq:ub}.
For the sake of simplifying the exposition throughout we will assume that \eqref{eq:ub} has a feasible solution, and
that the functions $f(\cdot), g_j(\cdot)$, $j=1,\dots,m$ are \change{Borel measurable in  $\D \subseteq \R$ \citep[cf.,][]{zulu05}}. 
Now let $J\subseteq \D$ be a set of given {\em atoms}, and construct the following linear program (LP) related to \refer{eq:ub} by associating a
probability decision-variable $p_x$ for every $x\in J$:
\beq \label{eq:mj}
\begin{array}{lllllllllllll}
	M_J^* :=~& \max_{p_x} & & \dsum_{x \in J} p_x f(x) \\
	& \text{s.t.} & & \sigma_j^- \leq \dsum_{x \in J} p_x g_j(x) \leq \sigma_j^+ & & j = 1, \ldots, m,\\
	& & & \dsum_{x \in J} p_x =1,  & &\\
	& & & p_x \geq 0 & & \text{for all $x \in J$} \\
\end{array}
\eeq
Furthermore, we assume that the set $J \subseteq \D$ is {\em feasible}; that is, the corresponding LP \refer{eq:mj} is feasible. The existence of such $J \subseteq \D$ follows from the classical result by \citet[Theorem 1]{kemp}, and can be found by solving algorithmically a {\em Phase I} version \citep[cf.,][]{bertlp} of the CG Algorithm~\ref{fig:alg}. 
Following CG terminology, given a set $J \subseteq \D$ we will refer to \refer{eq:mj}  as the {\em master} problem.

\change{Notice that any feasible solution of problem \eqref{eq:mj} will be a feasible (atomic distribution) for problem~\refer{eq:ub}. Also, the objectives of the two problems are the expected value of the function $f(x)$ over the corresponding decision variable distribution. Thus, $M_J^*$ is a lower bound for the optimal value of the upper semiparametric bound problem~\refer{eq:ub}.}
Furthermore, it is possible to iteratively improve this lower bound by updating the set $J \subseteq \D$ using the optimal {\em{dual values}} \citep[cf.,][]{bertlp} of the constrains after the solution of the master problem~\refer{eq:mj}. 
Namely, let $\rho_j^-,\rho_j^+$, $j = 1, \ldots, m$ and $\tau$ be the dual variables of the upper/lower moment 
(i.e., first set of constraints in eq.~\eqref{eq:mj}) and total probability 
(i.e, $\sum_{x \in J} p_x =1$)
constrains
 respectively. Given a feasible set $J \subseteq \D$, the dual variables can be used to select a new point $x \in \D$, to add to $J \subseteq \D$, that will make the corresponding LP \refer{eq:mj} a tighter approximation of \eqref{eq:ub}. In particular, given $\rho_j^-,\rho_j^+$, $j=1,\dots,m$ and $\tau$, consider the following {\em subproblem} to find $x$.
\beq \label{eq:srt}
\begin{aligned}
	S_{\rho, \tau}^* :=~& \underset{x \in \D}{\max} & & f(x) - \tau - \sum_{j=1}^m (\rho_j^+ + \rho_j^-) g_j(x) \\
\end{aligned}
\eeq

The objective value of \eqref{eq:srt} represents the {\em reduced cost} of adding the new point $x$ to $J$; that is, the marginal amount by which the objective in \refer{eq:mj} can be improved with the addition of $x$ in the master problem~\refer{eq:mj}. Using the master problem \eqref{eq:mj} and  subproblem~\eqref{eq:srt} admits an iterative algorithm that (under suitable conditions) converges to the optimal value of \eqref{eq:ub}. More specifically, at each iteration, the master problem \refer{eq:mj} is solved and 
its corresponding dual variables are used in the subproblem~\eqref{eq:srt} to select a new point $x$ to be added to 
the set $J \subseteq \D$. This is called a CG algorithm 
since at each iteration, a new variable $p_x$, corresponding to the new given point $x$,  is added to the master problem~\refer{eq:mj}.

If the functions $f(\cdot)$, $g_j(\cdot), j=1,\dots,m$ are continuous, and the support $\D$ of the 
underlying risk distribution is known to be compact, then the asymptotic convergence of the column generation algorithm follows from \citet[Thm. 5, Chp. 24]{Dantzig}. However, for the numerical solution of the 
practical instances of~\refer{eq:ub} considered here, it suffices to have a ``stopping criteria'' for
the CG algorithm.

\begin{theorem} \label{thm:conv}
\change{
Let $J \subseteq \D$ be given, and
$B^*$, $M^*_{J}$, $S^*_{\rho, \tau}$ be the optimal objective values of \eqref{eq:ub}, \eqref{eq:mj}, and~\eqref{eq:srt} respectively. 
Then $0 \leq B^* - M^*_{J} \leq S^*_{\rho, \tau}$.
}
\end{theorem}
Theorem~\ref{thm:conv} 
follows from \citet[Thm. 3, Chp. 24]{Dantzig}, and 
states that the LP approximation from below in \eqref{eq:mj} will be within 
$\epsilon$ of the optimal objective of \eqref{eq:ub} if the objective of subproblem \eqref{eq:srt} is less
than $\epsilon$. \change{It is worth mentioning that under additional assumptions about the feasible set of~\eqref{eq:ub}, one has that in the long-run $S^*_{\rho, \tau} \to 0$; that is, the CG algorithm will converge to the optimal solution of  the semiparametric bound problem \eqref{eq:ub}. This follows as a consequence of \citet[Thm. 5, Chp. 24]{Dantzig}.}
In practice, Theorem~\ref{thm:conv} provides a stopping criteria for the implementation of the CG algorithm under
only the assumption of the original
problem~\eqref{eq:ub} being feasible. Specifically, the CG Algorithm~\ref{fig:alg} can be used to find the optimal upper bound $B^*$ up to $\epsilon$-accuracy. As mentioned before, a {\em Phase I} version \citep[cf.,][]{bertlp} of the CG Algorithm~\ref{fig:alg},
can be used to construct an initial feasible set $J_0 \in \D$.


\begin{algorithm}[!htb]
\caption{Semiparametric bounds via column generation}
\begin{algorithmic}[1]
\Procedure{GC}{feasible $J_0$, $\epsilon > 0$}
\State $J \gets J_0$, $S_{\rho, \tau}^* = \infty$
\While{$S_{\rho, \tau}^* > \epsilon$}
\State {\bf compute} $M_{J}^*$, $p^* := \{p^*_x\}_{x\in J}$, the optimal objective and solution of master problem~\eqref{eq:mj}
\State  {\bf compute} $S_{\rho, \tau}^*$, $x^*$ , the optimal objective  and solution of
subproblem \eqref{eq:srt}
\State  $J \gets J \cup \{x^*\}$
\EndWhile
\State  {\bf return}  $J^* = J$, $p^*$, and $M_{J}^* \approx B^*$ (where $\approx$ stands for $M_{J}^*$ approximates $B^*$ )
\EndProcedure
\end{algorithmic}
\label{fig:alg}
\end{algorithm}


Note that (in principle) problem \refer{eq:ub}  has an infinite number of columns (i.e., variables) and a finite number of constrains;
that is, the semiparametric bound problem \eqref{eq:ub} is a {\em semi-infinite} program. Thus, 
the approach outlined above is an application to a semi-infinite program 
of column generation techniques
initially introduced by \citet{DantW61} for LPs, generalized
linear programs, and convex programs (cf. \citet[Chp. 22--24]{Dantzig}).
For a survey of column generation methods see \cite{colgen}.

\subsection{Solving the subproblem}
\label{sec:sub}

As observed by \citet{BirgD91}, the main difficulty in using the CG approach to solve the
semiparametric bound problem~\refer{eq:ub} is that the subproblem~\refer{eq:srt} is in general
a {\em non-convex} optimization problem \citep[cf.,][]{Nocedal2006NO}. However, in the
practically relevant instances of the problem considered here, the following assumption holds.

\begin{assumption}
\label{assume}
The functions $f(\cdot)$, and $g_j(\cdot)$, $j=1,\dots,m$ in \refer{eq:ub} are piecewise polynomials of degree less than five (5).
\end{assumption}

More specifically, typically no higher than fifth-order moment
information on the risk will be assumed to be known (e.g., $g_j(X) = X^j$ for $j=1,\dots,m$ and $m \le 5$). Also, the function $f(\cdot)$ typically defines: the piecewise linear cost or payoff of an  insurance instrument (e.g., $f(x) = \max\{0,x-d\}$); a ruin event using a (piecewise constant) indicator function (e.g., \change{$f(x) = \mathbb{I}_{[0,r]}(x)$}) ; or a lower than fifth-order moment of the risk or insurance policy cost/payoff \citep[see, e.g.][]{brock,cox, lo, schepper}. In such cases, the objective of \eqref{eq:srt} is an univariate
fifth-degree (or lower) piecewise polynomial. Thus, 
subproblem~\eqref{eq:srt}
can be solved ``exactly'' by finding the roots of fourth degree polynomials (using {\em first-order} optimality conditions \citep[see, e.g.,][]{Nocedal2006NO}). As a result, we have the
following remark.

\begin{remark}
\label{rem:LPs}
Under Assumption~\ref{assume} and using the CG Algorithm~\ref{fig:alg}, the solution to problem \refer{eq:ub} can be found by solving a sequence of LPs \refer{eq:mj} where the column updates \refer{eq:srt} can be found with simple arithmetic operations.
\end{remark}

\change{Moreover, thanks to current numerical algorithms for finding roots of univariate polynomials, it is not difficult to solve subproblem~\eqref{eq:srt} numerically to a high precision, even when the polynomials involved in the problem have degree higher than 5. In turn, this means that Algorithm~\ref{fig:alg} would perform well for instances of the problem with high degree polynomials.}

Generating semiparametric bounds using the CG approach outlined above has several key advantages over the semidefinite programming (SDP) solution
approach introduced by \citet{bert02, popescu}. First, only a linear programming solver for \eqref{eq:mj} and the ability to find the roots of polynomials with degree no more than four for \eqref{eq:srt} is required in most practically relevant situations. This means that the methodology can use any commercial LP solver allowing for rapid and numerically stable solutions. Second, the problem does not need to be reformulated for changes in the support $\D$ of the underlying risk distribution $\pi$ of $X$. Accounting for alternate support requires only limiting the search space in the subproblem \eqref{eq:srt}. Finally, problem~\eqref{eq:mj} is explicitly defined in terms of the distribution used to generate the bound value. So, for any bound computed, the {\em worst-case} (resp. {\em best-case} for the lower bound) distribution that yielded that bound can also be analyzed; with the SDP approach no such insight into the resulting distribution is, to the best of our knowledge, readily possible. The ability to analyze the resulting distribution would be of particular use to practitioners in the insurance and risk management industry and will be further discussed in Section~\ref{sec:worstcase}. Third, the CG approach works analogously for both the upper and lower semiparametric bound problems. In contrast, the SDP approach 
commonly results in SDP formulations of the problem that are more involved for the lower than for the upper semiparametric bound. Finally, as shown in Section~\ref{sec:extend}, the CG approach allows the addition of information about the class of distribution to which the underlying risk belongs (e.g., continuous, unimodal) without changing the core of the solution algorithm.

\ssection{Additional Distribution Information} \label{sec:extend}
As mentioned earlier, in practical instances of the semiparametric bound problem~\refer{eq:ub} the functions $g_j(\cdot)$, $j=1,\dots,m$ are typically set to assume the knowledge of moments of the underlying loss distribution; for example, by setting $g_j(X) = X^j$, $j=1,\dots,m$ in \refer{eq:ub}.
The general semiparametric bound problem \refer{eq:ub} can be extended to include additional distributional information other than moments \citep[see, e.g.,][]{popescu,schepper}. This is important as the resulting bounds will be tighter and the corresponding worst/best-case
distribution will have characteristics consistent with the practitioner's application-specific knowledge about continuity, unimodality, and heavy tails in financial loss contexts. In this section it is shown that 
the CG solution approach outlined in Section~\ref{sec:method} for semiparametric bound problems can
be extended to constrain the underlying distribution to 
be unimodal and continuous.

\subsection{Mixture Transformation} \label{sec:transform}

Note that a point  $x\in J^*$ obtained after running the CG Algorithm~\ref{fig:alg} can be interpreted
as the mean of Dirac delta distributions $\delta_{x}$ parametrized by (centered at) $x$. 
In turn, the resulting optimal distribution $\pi^*$ of the random variable $X$ in \refer{eq:mj} is a mixture of Dirac delta distributions;
that is, 
$\pi^* \sim \sum_{x \in J^*} p_{x} \delta_{x}$.
As we show below, the CG Algorithm~\ref{fig:alg} can be used to obtain optimal worst/best-case distributions
associated with the semiparametric bound problem~\refer{eq:ub} when besides the expected value constrains, information 
is known about the class of distribution to which $\pi$ belongs; for example, unimodal, smooth, asymmetric, etc. Basically
this is done by replacing $\delta_x \to H_x$ in the mixture composition of $\pi$, where $H_x$ is an appropriately chosen
distribution parametrized by $x$.

More specifically, assume that besides the moment information used in the definition of the semiparametric
bound problem~\refer{eq:ub}, it is known that the distribution $\pi$ is a mixture of
known probability distributions $H_x$, parametrized by a single parameter $x \in \R$.
For example, $x$ could be the mean of the distribution $H_x$, or $H_x$ could be a
uniform distribution between $0$ and $x$. Note that for any $g : \mathbb{R} \rightarrow \mathbb{R}$, 
it follows from the mixture composition of the distribution $\pi$ that:
\begin{align}
	E_{\pi}(g) &= \int_0^\infty g(u) E_{\pi(X)}(H_X(u) ) du \nonumber \\
	&= E_{\pi(X)} \left ( \int_0^\infty g(u) H_X(u) du\right ) \nonumber \\
	&= E_{\pi(X)}\left ( E_{H_X(U)}(g(U)) \right )\label{eq:mixh}.
\end{align}
This means that with the additional distribution mixture constrain, the associated semiparametric
bound problem can be solved with the CG Algorithm~\ref{fig:alg}
after replacing 
\begin{align}
f(x) &\to \int_0^{\infty} f(u) H_x(u)du = E_{H_x}(f),\nonumber \\
g_j(x) &\to  \int_0^{\infty} g_j(u) H_x(u)du =  E_{H_x}(g_j)\label{eq:trans},
\end{align}
for $j=1,\dots,m$ in \refer{eq:mj}, and \refer{eq:srt}. 

Note that in many instances, the expectations
in \refer{eq:trans} can be computed in closed-form as a function of the mixture
distribution parameter $x$. Moreover, the expectations in~\refer{eq:trans} are commonly
piecewise polynomials in $x$ (e.g., if $H_x$ is a uniform distribution between $0$ and $x$),
or can be written as polynomials after an appropriate change in variable
(e.g., if $H_x$ is a lognormal distribution with chosen volatility parameter $\sigma \in \R^+$, and mean $e^{x+\frac{1}{2}\sigma^2}$). In such cases, after applying the transformation \refer{eq:trans} the 
objective of subproblem~\refer{eq:srt} will be a piecewise polynomial on $x$. As discussed
in Section~\ref{sec:sub}, the subproblem can then be solved ``exactly'', 
and Remark~\ref{rem:LPs} will still hold as long as Assumption~\ref{assume} is valid after
the transformation~\refer{eq:trans}.
This is illustrated in Example~\ref{ex} below. Also, as we show with numerical experiments in Section~\ref{sec:num}, this is the case in most
practical applications.

\change{
\begin{example}
\label{ex}
Consider a simple insurance policy with  no deductilble 
on a loss $X$ for which the non-central  moments up to
$m$-order are assumed to be known. Specifically,  let $f(x) = \max\{0,x\}$, and $g_j(x) = x^j$, $j=1,\dots,m$. Also, assume that the distribution of the loss $X$ is known to be a mixture of uniform distributions $H_x$ of the form $H_x \sim \operatorname{Uniform}(0,x)$ in \refer{eq:mixh}. 
That is, $H_x(u) = \frac{1}{x}\mathbb{I}_{[0,x]}(u)$. From~\refer{eq:trans}, it follows that $E_{H_x}(f) = \frac{x}{2}$,
$E_{H_x}(g_j) = \frac{1}{j+1}x^{j}$ for $j=1,\dots,m$, and Remark~\ref{rem:LPs} will hold for any $m \le 4$.
\end{example}
}

In other cases; that is, when Assumption~\ref{assume} does not hold after the transformation~\refer{eq:trans}, one can sharply approximate the expectations
in \refer{eq:trans} using up to fifth-degree piecewise polynomials on~$x$ to take advantage of
Remark~\ref{rem:LPs}. Alternatively, given that the subproblem~\refer{eq:srt} is an univariate optimization problem,
global optimization solvers such as {\tt BARON} (cf., \cite{ts:05}) can be used to effectively solve it.

In Section~\ref{sec:unimod} we will discuss how the mixture transformation \refer{eq:trans}
can be used to substantially strengthen semiparametric bounds by using reasonable
assumptions about the 
underlying risk distribution regarding unimodality, and or continuity by using
a mixture of appropriate distributions. Moreover, in Section~\ref{sec:worstcase}, we use this transformation
to construct reasonable worst/best-case distributions associated to a given semiparametric
bound problem.

\subsection{Unimodality}
\label{sec:unimod}
\change{In many instances of the semiparametric bound problem, it might be reasonable to assume that the 
unknown distribution $\pi$ of $X$ in \refer{eq:ub} is unimodal with known mode \mode. This is particularly the
case 
when 
the underlying random variable  represents a 
financial asset or a portfolio of financial assets which are typically modelled by
a lognormal distribution when using parametric techniques \citep[see, e.g.][]{schepper}. 
In this section, we discuss how the unimodality information can be used in a straightforward fashion within the CG algorithm solution approach to obtain tighter semiparametric bounds. 
Before discussing this below, it is worth mentioning that there are problems in the context to Actuarial Science where
it is not appropriate to assume unimodality. For example, as shown
in \citet{LeeL10}, this is the 
case when the underlying random variable is associated with Property/Casualty Losses which often exhibit a multimodal
behaviour due to the combination sources compounding the loss (e.g., fire, wind, storms, hurricanes).}

It has been shown by \cite{popescu} that semiparametric bound problems with the additional constrain of the underlying distribution being unimodal can also be reformulated as a SDP by calling upon the classical probability result 
by \cite{khintchine} regarding the representation of unimodal distributions. 
Specifically, \cite{khintchine} proved that any unimodal distribution can be represented by a mixture of uniform distributions each of which have \mode as an endpoint (either the right or left endpoint). This same result can be embedded in the framework of Section~\ref{sec:method} by leveraging the variable transformation of Section~\ref{sec:transform}.

Recall that the CG algorithm can be defined in terms of mixing distributions $H_x$,
where $x$ represents a parameter of the distribution. In particular, for given (mode) \mode$\in \R$, let
\beq \label{eq:unifh} H_x \sim \Unif(\min\{x,\mode\}, \max\{x,\mode\}). \eeq
Using $H_x$ above in \refer{eq:trans} to transform the semiparametric bound problem~\refer{eq:ub} will lead to a bound
over distributions that are unimodal with mode $\mode$. 

The simple transformation \refer{eq:trans} using the mixture of uniforms \refer{eq:unifh} allows the CG approach to leverage the results of \cite{khintchine} while avoiding a complex reformulation as in the case of the SDP methodology of \cite{popescu}. Enforcing unimodality is a straightforward special case that highlights the flexibility of the methodology discussed in Section~\ref{sec:method}.

\subsection{Smoothness and Unimodality}
\label{sec:cont}
The base method of Section~\ref{sec:method} computes the desired semiparametric bounds, and provides a discrete (atomic) worst/best-case distribution $(x,p_x)$ for all $x \in J$ associated with the bound.
In practice it is more desirable and intuitive to work with a continuous probability density function. If one is considering a problem measuring financial loss, then having discrete loss values may not provide the insight that a continuous probability density function would, given that a discrete collection of outcomes is highly unrealistic.
Using the uniform mixture defined in \eqref{eq:unifh} is guaranteed to yield a unimodal distribution in the computation of the semiparametric bounds \refer{eq:ub}. However, the resulting density will contain multiple discontinuities including at the mode itself. Furthermore, the density will only be nonzero over the interval $[\min\{x: x \in J^*\},\max\{x: x \in J^*\}]$; that is, it has finite support. 
It would be desirable to obtain worst/best-case distributions associated with the semiparametric bounds that are {\em smooth}; that is, both continuous and differentiable.

Below we show that by appropriately choosing
the distribution $H_x$ (and its parameters) in the mixture, it is possible to 
obtain worst/best-case distributions that are both smooth and unimodal, and that closely replicate the corresponding upper (best) and lower (worst) semiparametric bounds.
This can be readily done using the CG approach by reformulating the 
semiparametric bound problem~\refer{eq:ub} using the transformation~\refer{eq:trans}, and choosing 
\begin{equation}
\label{eq:logmix}
H_x \sim \operatorname{lognormal}(\mu_x, \sigma_x),
\end{equation}
where $\mu_x, \sigma_x$ are given in terms of $x$ by the equations
\beq
\label{eq:lnparam}
\begin{array}{rl}
e^{\mu_x+ \frac{1}{2}\sigma_x^2} & = x,\\
(e^{\sigma_x^2}-1)e^{2\mu_x +\sigma_x^2} & = \alpha^2.\\
\end{array}
\eeq
for a given $\alpha \in \R^+$. That is, the lognormal distribution $H_x$ is set to have a mean of $x$, and variance~$\alpha^2$.
Note that besides the mean parameter $x$, which will be used to construct the mixture using the CG Algorithm~\ref{fig:alg}, one needs to set a second parameter $\alpha$ in \refer{eq:lnparam} to properly define the 
lognormal distribution $H_x$ in \refer{eq:logmix}.

The  lognormal mixture approach (i.e., \refer{eq:logmix}, and \refer{eq:trans}) can be used to obtain solutions to the semiparametric bound problem \refer{eq:ub} where the underlying worst/best-case distribution is both unimodal, and smooth,
and replicates as close as possible the semiparametic bound obtained when the distribution is assumed to be unimodal (and not necessarily smooth). This is in part thanks to the additional degree of freedom given by the choice of the parameter $\alpha$ in \refer{eq:lnparam}.
\change{To see this, let us refer to 
\beq
\label{eq:moms}
\begin{array}{l}
\tilde{\mu}^+ := \displaystyle \sup_{\pi \text{ a p.d. on $\D$}} \{E(X)~|~\sigma_j^- \le E_{\pi}(g_j(X)) \le \sigma_j^+, j=1,\dots,m\},\\
\tilde{\mu}^- :=  \displaystyle \inf_{\pi \text{ a p.d. on $\D$}} \{E(X)~|~\sigma_j^- \le E_{\pi}(g_j(X)) \le \sigma_j^+, j=1,\dots,m\},\\
\tilde{\sigma}^2 :=  \displaystyle \sup_{\pi \text{ a p.d. on $\D$}} \{\operatorname{Var}(X)~|~\sigma_j^- \le E_{\pi}(g_j(X)) \le \sigma_j^+, j=1,\dots,m\},\\
\end{array}
\eeq
and assume that $\tilde{\mu}^+, \tilde{\mu}^-, \tilde{\sigma}^2$ are bounded (i.e., this 
will be the case if $g_j(X) = X^j$, $j=1,\dots,m$, with $m \ge 2$ in \refer{eq:ub}), and that
$\D \subseteq \R^+$ (as in practice). 
Clearly, for the lognormal distribution mixture~\refer{eq:logmix} to be feasible for the semiparametric bound problem~\refer{eq:ub}, $\alpha$ in \refer{eq:lnparam} should be chosen such that $\alpha^2 \le \tilde{\sigma}^2$ to ensure that the variance of the lognormal distribution used for the mixture is less than the
maximum possible variance of the probability distributions $\pi$ associated with the expected value constrains in \refer{eq:ub}.  
Moreover, as $\alpha \to \tilde{\sigma}$. That is, the only feasible solution of the semiparametric bound problem with lognormal mixture would be a (single) lognormal with variance $\alpha^2$, and mean $x$ satisfying $\tilde{\mu}^- \le x \le \tilde{\mu}^+$, which is unimodal. Thus, there exists an $\alpha \in [0,\tilde{\sigma}]$, such that the lognormal mixture obtained with the CG Algorithm~\ref{fig:alg} will be unimodal. To find the value of $\alpha$ such that the lognormal mixture obtained with the CG approach is both unimodal and as close as possible to replicate the semiparametric bound obtained by assuming that the probability distribution $\pi$ in \refer{eq:ub} is unimodal (and not necessarily smooth), one can use the {\em bisection} Algorithm~\ref{fig:unimodbis}.}

\begin{algorithm}[!htb]
\caption{Smooth and unimodal worst/best-case distribution}
\begin{algorithmic}[1]
\Procedure{Bisection}{$0 < \alpha_{lo} < \alpha_{hi} < \tilde{\sigma}$, $\epsilon >0$}
\While{$|\alpha_{hi} - \alpha_{lo}| > \epsilon$}
\State $\alpha_{k} \gets \frac{1}{2}(\alpha_{lo}+\alpha_{hi})$
\State {\bf compute} $J^*$, $p^* := \{p^*_x\}_{x\in J}$, using CG Algorithm~\ref{fig:alg} and $H_x$ in \refer{eq:logmix}
\If{$\pi \sim \sum_{x \in J^*} p^*_x H_x$ is unimodal} \State{$\alpha_{hi} \gets \alpha_k$}
\Else  \State{$\alpha_{lo} \gets \alpha_k$} \EndIf
\EndWhile
\State  {\bf return} $J^*$, $p^*$, $\alpha = \alpha_k$, and $M_{J}^*$
\EndProcedure
\end{algorithmic}
\label{fig:unimodbis}
\end{algorithm}

Note that in the discussion above, the choice of the lognormal distribution is not key. Instead, the same would apply as long as the mixture distribution in \refer{eq:logmix} is smooth, unimodal, and has at least two appropriate degrees of freedom in the choice of parameters (e.g., in case of random variables with support on the whole real line, the normal distribution could be used to form the mixture).
In Section~\ref{sec:worstcase}, we illustrate with a numerical example how the {\em bisection} Algorithm~\ref{fig:unimodbis} can
be used to obtain a smooth and unimodal worst-case distribution that closely replicates the behaviour of the worst-case unimodal distribution.

When using a smooth distribution to define the mixture component $H_x$ in \refer{eq:trans},
it is important to understand the impact of the selection of mixture components $H_x$. Ideally, computing the bound with a mixture of smooth distributions $H_x$ would yield the optimal value across all possible smooth distributions in the semiparametric bound problem~\refer{eq:ub}. Instead, it is the semiparametric bound across all mixtures with components $H_x$. However, in Theorem~\ref{thm:unifcv} below, we show that the optimal semiparametric bounds across all smooth unimodal distributions is the 
same as the one across unimodal distributions. Loosely speaking, this
follows from the fact that the density function of a uniform distribution can be 
arbitrarily approximated by an appropriate smooth density function.


\begin{theorem}
	\label{thm:unifcv}
	The semiparametric bound problem~\refer{eq:ub}, with the additional constrain of the underlying distribution being smooth and unimodal  is equivalent to problem~\refer{eq:ub}, 
with the additional constrain	 of the underlying distribution being unimodal.
\eqref{eq:unifh}.
\end{theorem}
\begin{proof}
Let $B^*_u$ be the bound corresponding to the semiparametric bound problem~\refer{eq:ub}, 
with the additional constrain	 that the underlying distribution $\pi$ is unimodal.
Note that there exists a distribution $\pi^*$ such
that  $B^*_u := E_{\pi^*}(f(X))$ and $\pi^*$ is a mixture of uniform distributions (cf., Section~\ref{sec:unimod}); that is, with $H_x\sim \Unif(\min\{x,\mode\}, \max\{x,\mode\})$ in \refer{eq:trans}. Now, for $\eta >0$, let $\pi^\eta$ be the mixture obtained after replacing $H_x$ by
\beq \label{eq:log} H^\eta_x(u) = \frac{1}{b(x)-a(x)} \left[\frac{1}{1+e^{-\eta (u-a(x))}} - \frac{1}{1+e^{-\eta (u-b(x))}}\right], \eeq
in the mixture $\pi^*$, where $a = \min\{x, \mode\}$ and $b = \max\{x,\mode\}$, where $\mode$ is the mode of $\pi^*$. The statement follows since by letting $\eta \to \infty$, one obtains a smooth distribution $H^\eta_x$ (see, Lemma~\ref{totpro} in Appendix~\ref{app:approx}) that is arbitrarily close to $H_x$ (see, Lemma~\ref{l:u} in Appendix~\ref{app:approx}) .
\end{proof}

A numerical example to illustrate Theorem~\ref{thm:unifcv} is provided in Section~\ref{sec:worstcase}.

\section{Numerical Illustration}
\label{sec:num}

Problem \eqref{eq:ub} is of particular interest in actuarial science because the target function $f(\cdot)$ in \refer{eq:ub} can take the form of payoffs for common insurance and risk management products for which the distribution of the underlying random loss is {\em ambiguous} \citep[see, e.g.,][for a recent reference]{Ye10, Nata11}; that is, it is not known precisely. Let $X$ represent the loss, and $d$ 
be the deductible associated with an insurance policy on $X$. For example, \citet[][Sections 3.1 and 3.2]{schepper}
provide upper and lower bounds on the expected cost per policy payment $\max(X-d,0)$ when only
up to third-order moment information on the loss distribution is assumed to be known. This is done by solving \refer{eq:ub}  analytically with 
\begin{equation}
\label{eq:ex2}
f(X) =\max(X-d,0), g_j(X) = X^j, \sigma_j^+ = \sigma_j^- \text{ for } j = 1,\dots,m,  \D = \R^+, \text{ and } m=2,3.
\end{equation}
In practice, losses do not exceed certain maximum, say $b$. Taking this into account, \citet[][Proposition 3.2]{cox} provides upper and lower bounds on the expected cost per policy $\max(X-d,0)$ when only the maximum potential loss $b$ and up to second-order moment information on the loss distribution is assumed to be known. Accordingly, this is done 
by solving \refer{eq:ub}  analytically with 
\begin{equation}
\label{eq:ex1}
f(X) =\max(X-d,0), g_j(X) = X^j,\sigma_j^+ = \sigma_j^- \text{ for } j = 1,\dots,m,  \D = [0,b] \subset \R^+, \text{ and } m=2.
\end{equation}

\subsection{Second-order LER Bounds with Unimodality} \label{sec:ex1m2}

Let us reconsider the semiparametric bound on the expected cost per policy defined in~\refer{eq:ex1}. Note that from semiparametric bounds on the expected cost per policy, one can readily obtain bounds on the expected {\em loss elimination ratio} (LER) of the policy. Specifically, note that the expected LER associated with a policy with payoff $\max\{0, X-d\}$ is \citep[cf.,][]{cox}
\beq \label{eq:lerdef} E({\rm LER}(X)) = \frac{E(\min(X,d))}{E(X)} =  \frac{E(X) - E(\max(X-d,0))}{E(X)}.\eeq
Being able to compute bounds on the expected LER would be beneficial for an insurer attempting to set a deductible in cases where the actual loss distribution is ambiguous. For example, 
in \citet[][Section 3]{cox}, the relationship \refer{eq:lerdef} is used to obtain upper and lower semiparametric bounds on the 
expected LER of an insurance policy with deductible, assuming only the knowledge of the mean and variance of the loss, and that the loss cannot exceed a known maximum. 

Another sensible premise is to assume that the loss distribution is unimodal. To illustrate the potential of the CG approach, in Figure~\ref{fig:ler}, we use the mixture transformation of Section~\ref{sec:unimod} to compute upper and lower semiparametric bounds on the insurance policy with payoff $\max\{0, X-d\}$ when the mean $\mu$ and variance $\sigma^2$ of the loss are assumed to be known, and the loss cannot exceed the value of $b$, where $d$ is the policy's deductible, and the loss distribution is also assumed to be unimodal with mode \mode. The results are compared with the analytical formula of \citet[][Section3]{cox} to illustrate the tightening of the bounds obtained by adding the unimodality assumption. Specifically, following \citet[Section 4]{cox}, in Figure~\ref{fig:ler} we set $\mu = 50$, $\sigma = 15$, $b=100$, and $\mode = \{45, 50\}$.

\begin{figure}[ht]
	\centering
	\includegraphics[width = 0.475\textwidth]{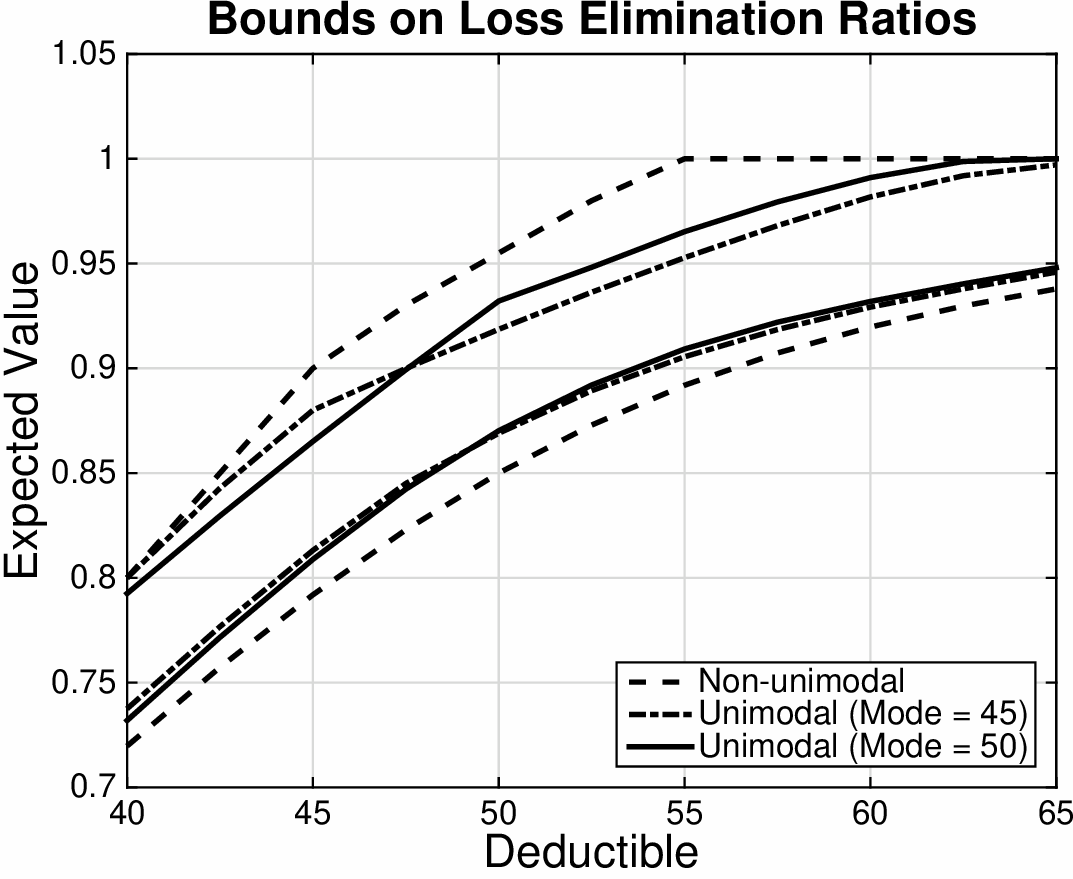} \includegraphics[width = 0.475\textwidth]{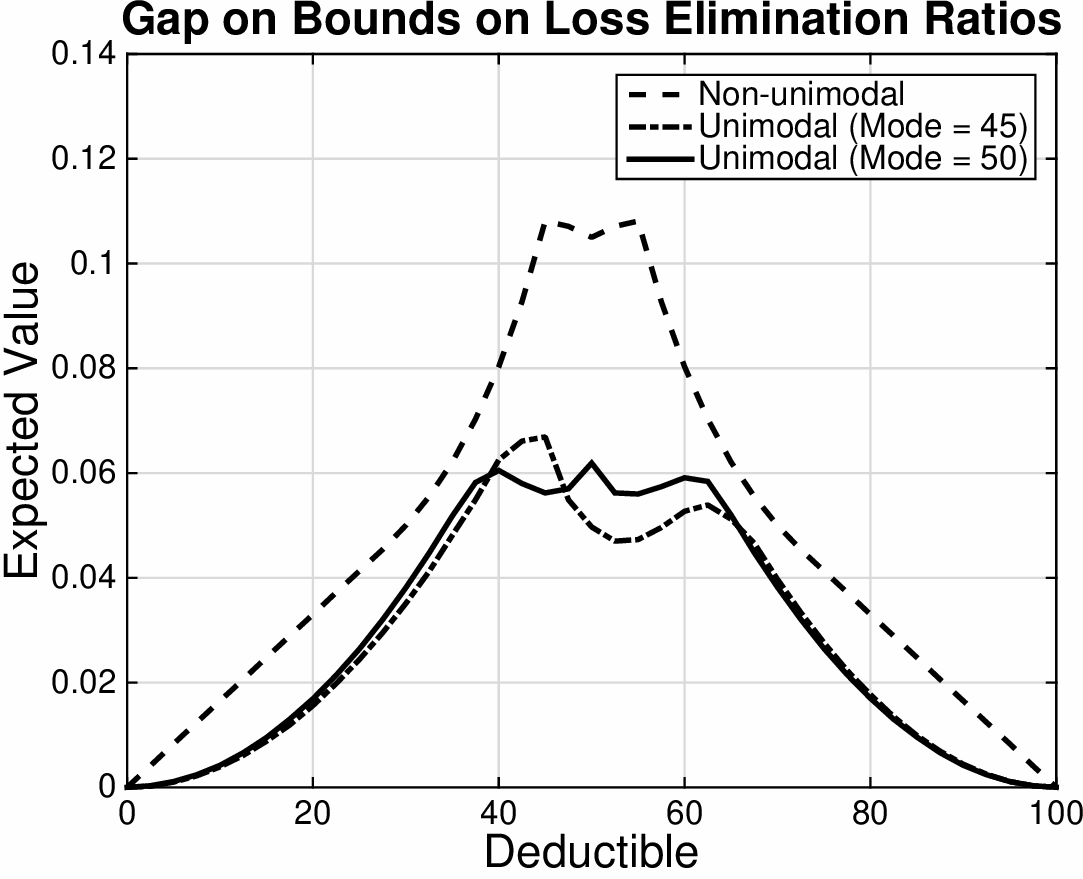} 
	\caption{Expected LER bounds (left) and gaps (right) for different values of the deductible $d$, when the mean $\mu = 50$, and variance 
	$\sigma^2 = 225$ of the underlying loss, as well as its potential maximum value $b=100$, are assumed to be known. Gaps indicates the difference between upper and lower bounds. Results are presented for bounds without the unimodality constrain, and with unimodality constrain with mode $\mode = \{45, 50\}$.}
	\label{fig:ler}
\end{figure}

Observe in Figure~\ref{fig:ler} that the expected LER {\em gap}; that is, the difference between the upper and lower semiparametric expected LER's bounds, is significantly tighter when unimodality is included. When $\mode = 50$ the gap is symmetrical with a small spike at the mean/mode. The case in which $\mode = 45$ yields a corresponding peak in the gap around the mode. 
For either very high or very low
deductible values, the choice of the mode has little impact on the size of the gap. 
Regardless of the mode's value,
the gaps are of similar magnitude and narrower than in the case where the underlying loss distribution is not assumed to be unimodal.

\subsection{Examining worst-case distributions} \label{sec:worstcase}

Suppose we wish to compute the semiparametric upper bound defined by \refer{eq:ex2} with $m=2$. 
Specifically, let the expected value (moment) constrains in \refer{eq:ub} be given by
\beq
\label{eq:const}
E_\pi(X) = \mu, E_\pi(X^2) = \sigma^2 + \mu^2,
\eeq
for given $\mu, \sigma \in \R^+$. That is, both the mean and the variance of the underlying loss distribution
are assumed to be known. In this case, equation \refer{eq:moms} reduces to
$\mu_\pi^+ = \mu_\pi^- = \mu$, and $\sigma_\pi^2 = \sigma^2$.

A closed-form solution for the corresponding semiparametric upper bound problem was derived by  \cite{lo}, where he considers $f$ in \refer{eq:ex2} as the payoff of an European call option with strike $d$. If furthermore, it is assumed that the underlying distribution $\pi$ of the loss (or asset price) is unimodal, the corresponding semiparametric upper bound can be computed using: 
the analytical formula provided by \citet[][Section 3.3]{schepper}, the SDP techniques provided by \citet{popescu}, 
or the uniform mixture  approach of Section~\ref{sec:unimod} with components \eqref{eq:unifh}. The CG uniform mixture method readily provides a worst-case distribution. This distribution, however, is not smooth, has finite support, and is unrealistic as a model for the uncertainty of losses. For this reason we compute upper bounds using smooth mixture compontents in \refer{eq:trans} and inspect the resulting worst-case probability and cumulative density functions. The resulting smooth distribution is then compared to the optimal unimodal uniform mixture distribution. Specifically, we use the lognormal mixture \refer{eq:logmix}. 

In particular, let us sample the values of $\mu, \sigma$ in \refer{eq:const} from a lognormal asset price dynamics model which is also commonly used to model (a non-ambiguous) loss distribution \citep[see, e.g.][]{Cox2004, Jaimungal2006}. Namely, let
$\mu = X_0e^{rT}$, and $\sigma= X_0e^{rT}(e^{\nu^2T}-1)^{\frac12}$ for values of $X_0 = 49.50$, $r = 1\%$, $\nu = 20\%$, and $T = 1$.
Also let $d = X_0$ in \refer{eq:ex2}; that is, consider a policy where the expected value of the loss is equal to the deductible.

The semiparametric upper bound was computed 
using the lognormal mixture \refer{eq:logmix} for different values of $\alpha \in [1, 1.5, \ldots, 20]$ and the percent above the parametric value of the policy based on Black-Scholes formula was plotted in Figure~\ref{fig:byalpha}. The corresponding semiparametric bound without the unimodality assumption (given by \cite{lo}), and unimodal bounds with an uniform mixture from Section~\ref{sec:unimod}   are also plotted for reference.
The bold point in Figure~\ref{fig:byalpha}
 represents the smooth, unimodal bound obtained with the bisection Algorithm~\ref{fig:unimodbis} and a mixture of lognormal distributions. Also, the plot illustrates the point ($\alpha = 11.34$) at which the the bounds obtained by the bisection Algorithm~\ref{fig:unimodbis} and the CG Algorithm~\ref{fig:alg} with a mixture of uniform distributions are equal.

In Figure~\ref{fig:byalpha} one can observe that for extremely low values of $\alpha$, the component distributions of the mixture are very narrow, approaching the pessimistic discrete distribution case associated with closed-form bound of \cite{lo}. We also see that as $\alpha \rightarrow \sigma = 20.4$ the resulting bound distribution converges to the lognormal specified by the Black and Scholes asset pricing framework. This convergence is seen in Figure~\ref{fig:byalpha} since the error goes to zero and the upper bound price equals the analytical Black-Scholes price.

\begin{figure}[ht]
	\centering
	\includegraphics[width = 0.5\textwidth]{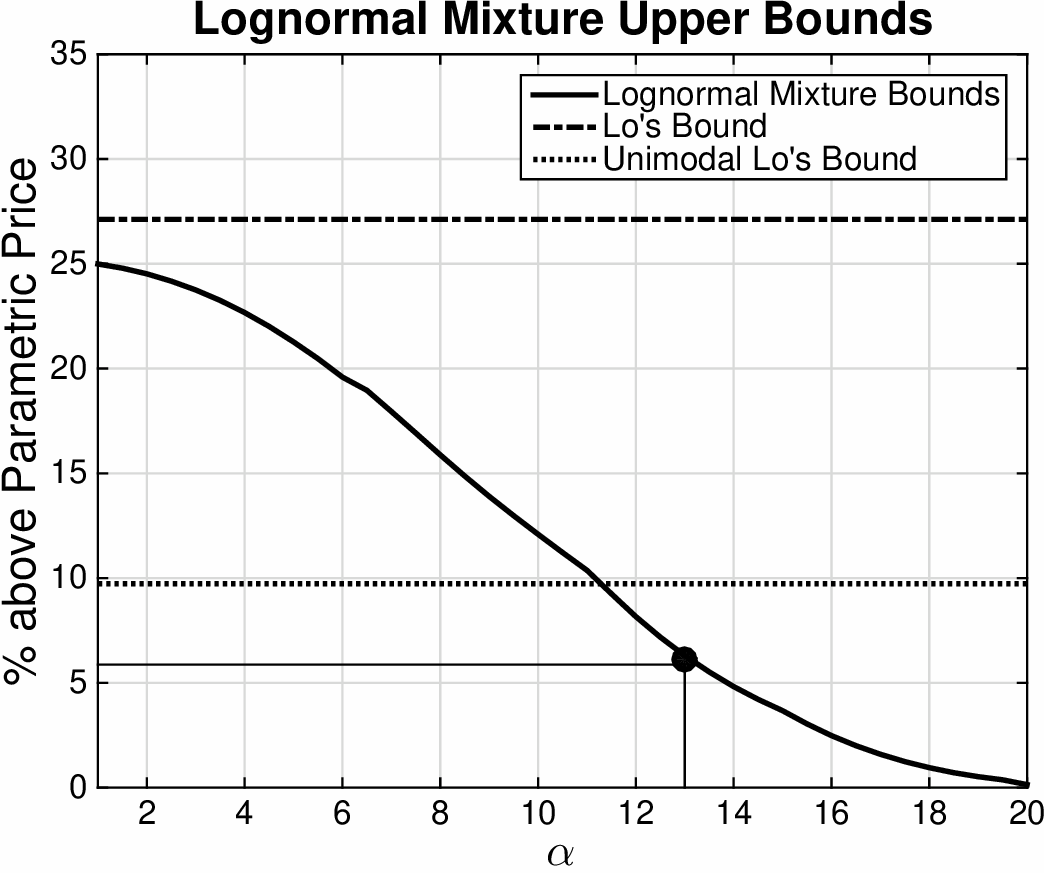}
	\caption{Percentage above the parametric Black and Scholes price of the \citet{lo} upper bound (Lo's Bound) and the lognormal mixtures obtained from Algorithm~\ref{fig:unimodbis} (lognormal Mixture Bounds).  The bold point denotes the value of $\alpha = 13.75$ in which the lognormal mixture bound obtained from Algorithm~\ref{fig:unimodbis} produces a unimodal distribution.} 
	\label{fig:byalpha}
\end{figure}

Figure~\ref{fig:byalpha} also highlights the result discussed in Theorem~\ref{thm:unifcv}. The bound computed using uniform mixture components is greater than the unimodal bound from the lognormal mixture with the gap size under 4\%. Note that the unimodal upper bound using lognormal mixture components occurred at $\alpha = 13.75$. As mentioned before, the $\alpha$ that yields the same bound value as that from the uniform mixture is $\alpha = 11.34$. The smaller the $\alpha$ in the lognormal mixture \refer{eq:logmix} the higher the conservatism associated with the semiparametric bound. Figure~\ref{fig:df} shows the probability distribution function (PDF) and cumulative distribution function (CDF) of the lognormal mixtures at $\alpha = \{11.34, 13.75\}$ as well as the associated {\em true} lognormal distribution with mean and variance equal to the moments used to compute the semiparametric bounds. To highlight the advantage of using the lognormal mixtures instead of the uniform mixtures, Figure~\ref{fig:pcdfpop} shows the optimal PDF and CDF of the latter along with the associated true lognormal distribution.

\begin{figure}[ht]
	\centering
	\includegraphics[width = 0.475\textwidth]{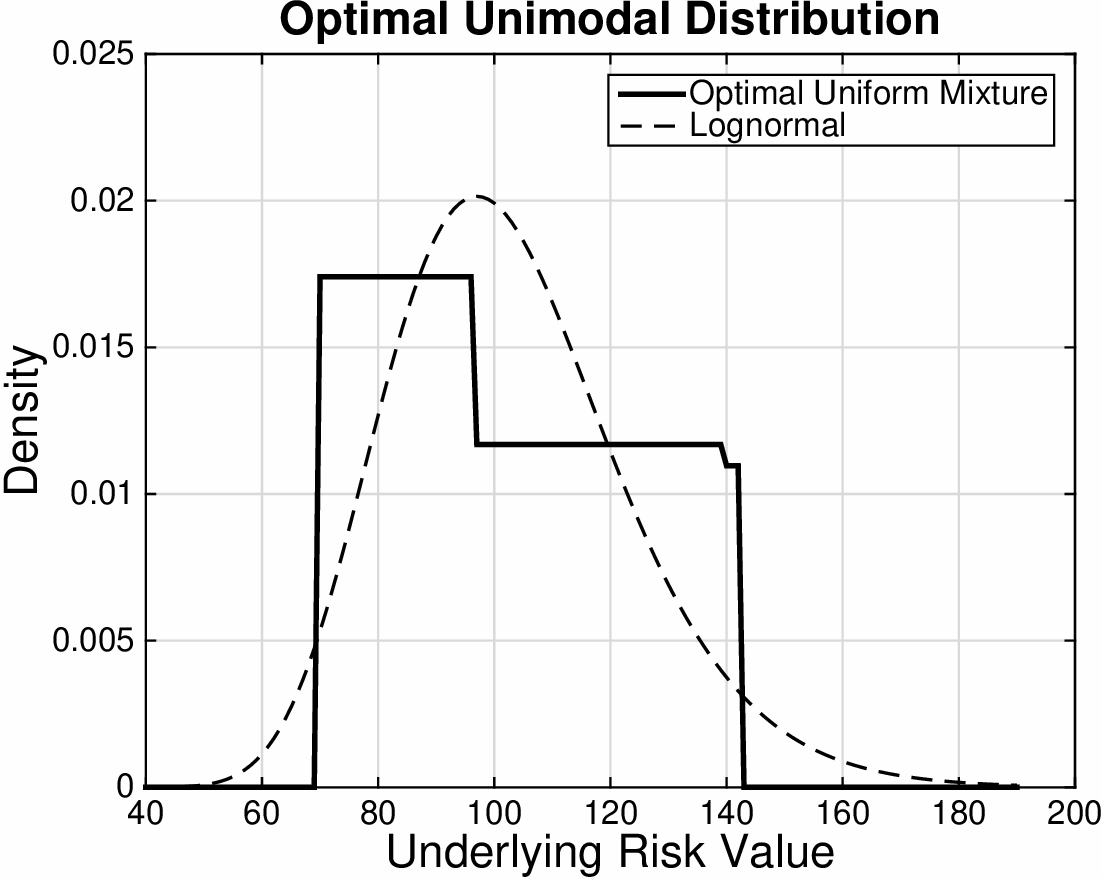} \includegraphics[width = 0.475\textwidth]{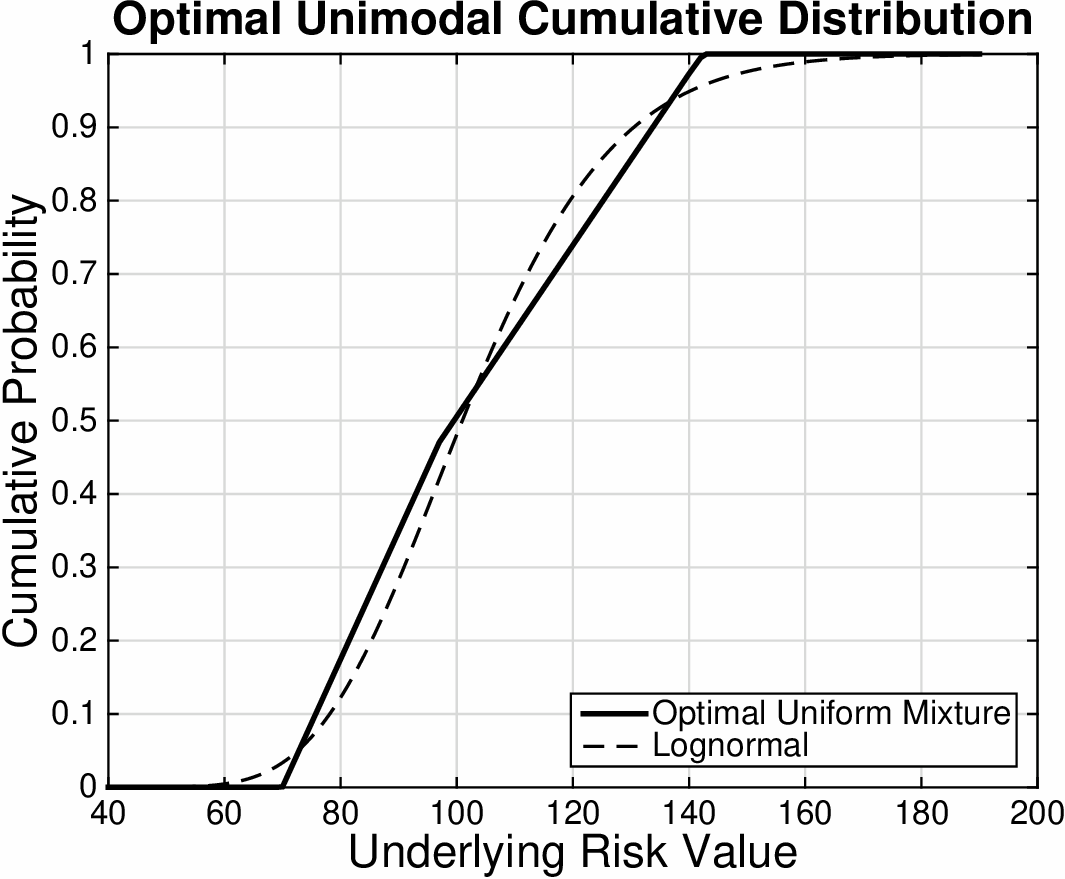}
	\caption{PDF and CDF that yields the optimal unimodal bound via uniform mixtures compared with an associated lognormal distribution.}
	\label{fig:pcdfpop}
\end{figure}

\begin{figure}[ht]
	\centering
	\includegraphics[width = 0.475\textwidth]{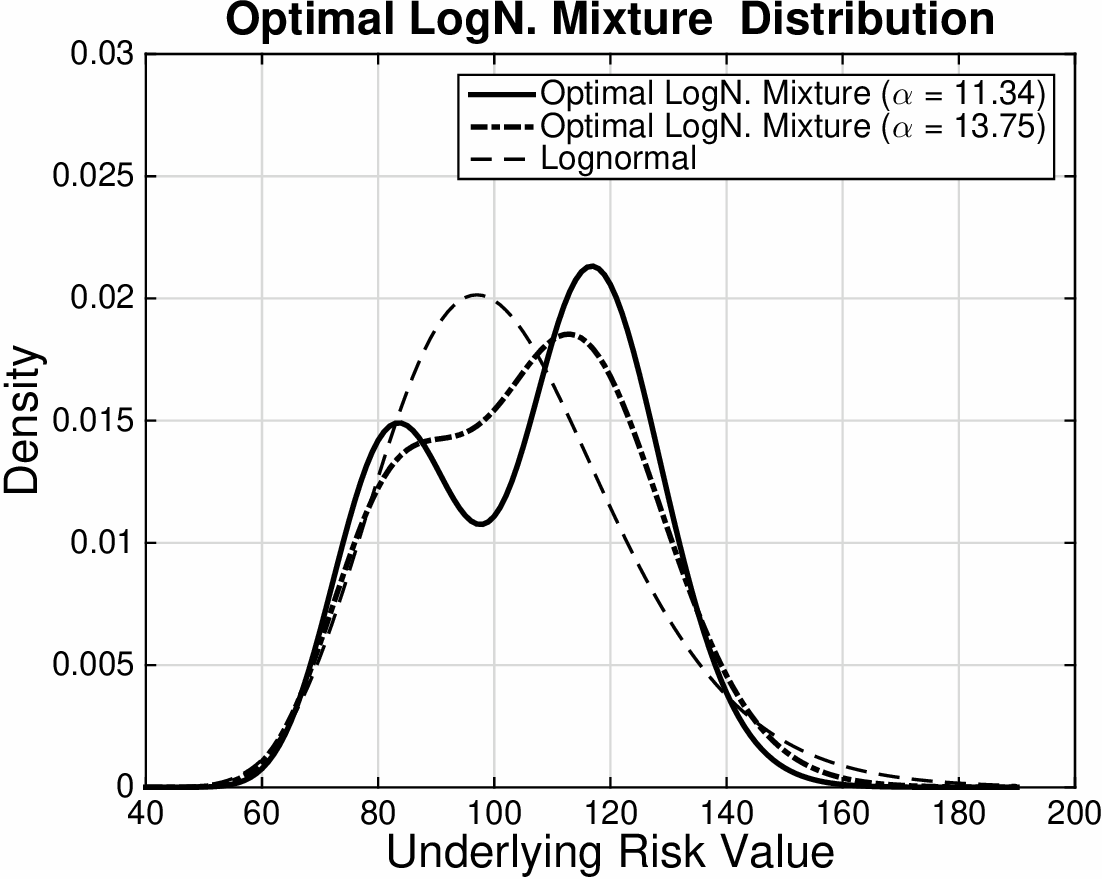} \includegraphics[width = 0.475\textwidth]{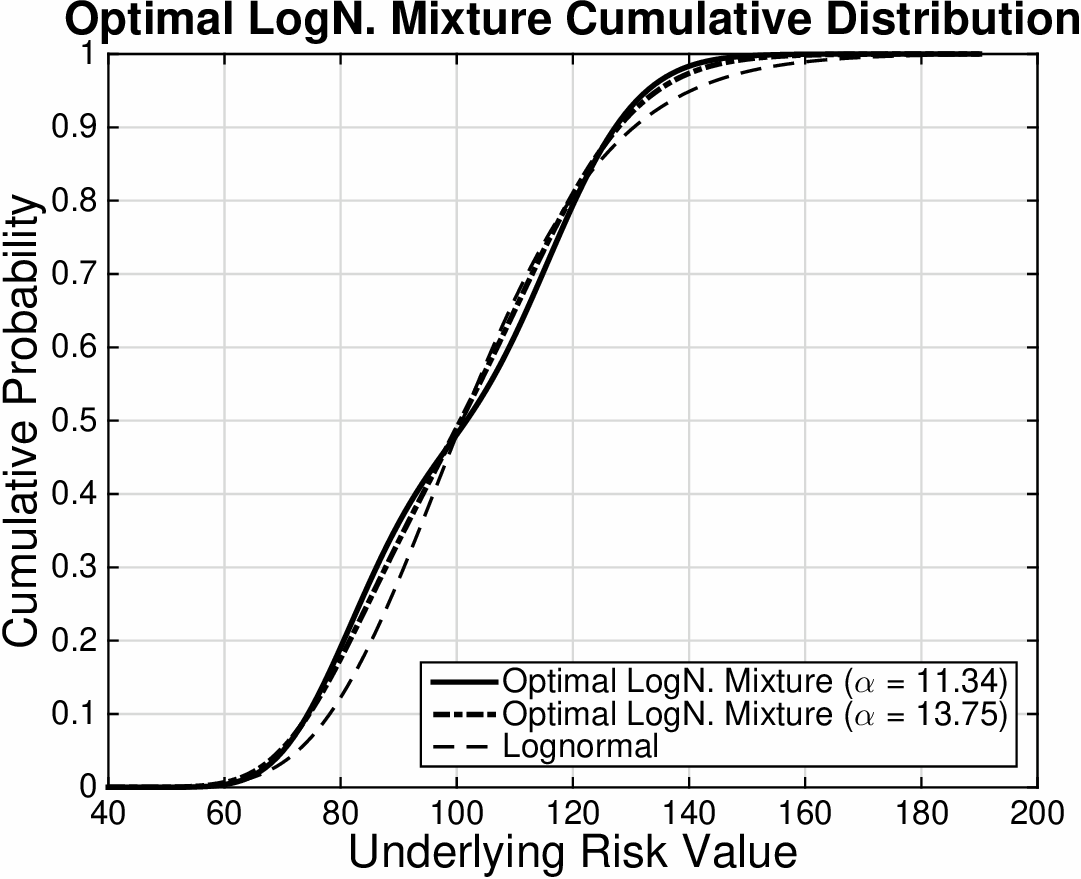}
	\caption{PDFs and CDFs that yield the optimal bounds  via lognormal mixtures (cf., Algorithm~\ref{fig:unimodbis}) for $\alpha = \{11.34, 13.75\}$, 
	 compared with an associated lognormal distribution.}
	\label{fig:df}
\end{figure}

Observe in Figure~\ref{fig:df} that the unimodal lognormal mixture at $\alpha = 13.75$ is relatively close to the shape of the associated true lognormal distribution. Contrast this with the PDF of the unimodal mixture of Figure~\ref{fig:pcdfpop} which bares little similarity to the associated true lognormal probability density. The primary difference to note is that the lognormal mixture approach yields an unimodal distribution, but the mode is not specified.
Using the uniform mixture approach of Section~\ref{sec:unimod} will produce a density with a specified mode, but at the cost of an unrealistic distribution. The lognormal mixture at $\sigma = 11.34$ is bimodal and does not resemble the true density. In each case the cumulative densities are fairly close to the true CDF. This example highlights the ability of the lognormal mixture approach to construct realistic unimodal distributions while still being close to the optimal unimodal bound; here the gap was shown to be under 4\%.

\subsection{Illustration of Theorem~\ref{thm:unifcv}}
We finish this section by providing numerical results to illustrate Theorem~\ref{thm:unifcv}. Reconsider the semiparametric bound problem defined in \refer{eq:ex2} with $m=2$ (i.e., with up to second-order moment information),  and the additional constrain that the underlying loss distribution is unimodal. 

Suppose we compute the semiparametric upper bound of the at-the-money policy described in Section~\ref{sec:worstcase} enforcing the first two known moments and continuity. To illustrate Theorem~\ref{thm:unifcv}, the upper bound is also computed for the option using \eqref{eq:log} and various levels of $\eta$. The percentage difference between the former and latter are plotted against $\eta$ in Figure~\ref{fig:unifconv}.

\begin{figure}[H]
	\centering
	\includegraphics[width = 0.5\textwidth]{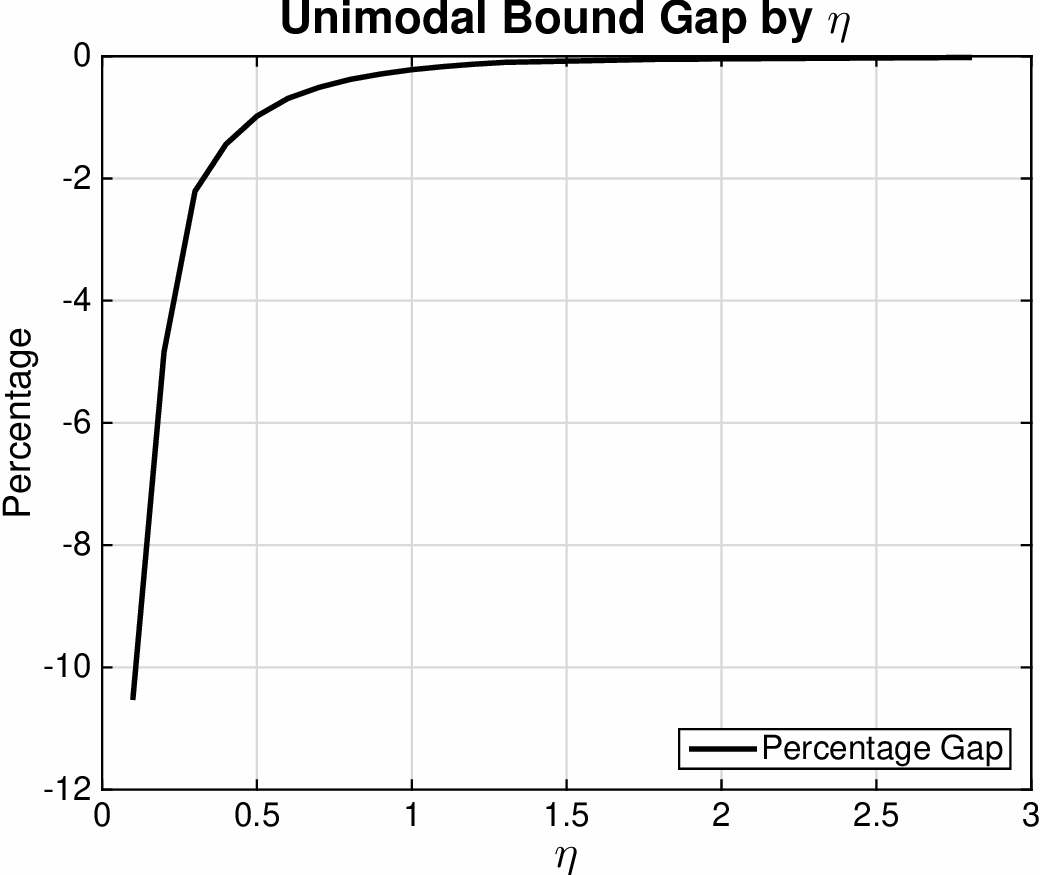}
	\caption{Illustration of how the upper bound with mixture components \eqref{eq:log} converges to the 
	unimodality bound~\eqref{eq:unifh} (without smoothness requirements)
	 as $\eta \rightarrow \infty$. The plot shows the difference in percentage between these bounds 
	 as a function of $\eta$.}
	\label{fig:unifconv}
\end{figure}

From Figure~\ref{fig:unifconv} we see that by implementing the algorithm with \eqref{eq:log} and increasing $\eta$ the upper bound converges to that computed with mixture component \eqref{eq:unifh}. Bounds computed using smoothness 
and unimodality can yield values arbitrarily close to, but not greater than those obtained when only unimodality is enforced. The implication of Theorem~\ref{thm:unifcv} is that any tightening of the upper bound from a smooth mixture is a byproduct of the choice of the mixture distribution $H_x$, not from the inclusion of smoothness. In practice it can be confirmed that the change in bound from choice of $H_x$ is generally fairly small.

\ssection{Extensions}
Besides the common features of insurance policies considered in Section~\ref{sec:num} such as the policy's deductible $d$, and the fact that losses typically do not exceed a maximum loss $b$; a maximum payment, and coinsurance are also common features in insurance policies. If a policy will only cover up to a maximum loss of $u \in R^+$, and coinsurance stipulates that only some proportion $\gamma \in [0,1]$ of the losses will be covered, then the policy's payoff can be written as $f(X) = \gamma[\min(X,u) - \min(X,d)]$. 
All of these policy modifications can be readily incorporated into the CG solution approach framework.
In particular, the CG methodology can be applied to compute bounds on the expected policy loss, for a wide variety of standard functions of loss random variables used in industry.

In the numerical examples in Section~\ref{sec:num}, the target function $f$ in \refer{eq:ub} was used to model piecewise linear insurance policy payoffs, and the functions $g_j$, $j=1,\dots,m$ to use the knowledge of up to m-order non-central moments of the loss distribution.
However,  
the methodology discussed here applies similarly to functional forms of $f$ that are not piecewise linear policy payoffs, and the 
 functions $g_j$, $j=1,\dots,m$ can be used to represent the knowledge of other than non-central moment's
 information. As an example, let $c_j$ be the European call prices on some stock $X$ for various strike prices $K_j$, $j=1,\dots,m$. Recall that the payoff for this type of option is the same as that of the $d$-deductible policy described in \eqref{eq:ex2}. The constrain set for \eqref{eq:ub} can be defined to enforce market prices of options by setting\beq \label{eq:price} E_\pi(g_j(X)) := E_\pi(\max(X-K_j,0)) = c_j, j=1,\dots,m. \eeq
The market price constrains \refer{eq:price} can then be used to compute semiparametric bounds on the variance of the underlying asset. This is accomplished by definining the target function $f(X)$ in \refer{eq:ub} as
\beq \label{eq:risk} \quad E_\pi(f(X)) := E_\pi((X-\mu)^2) \eeq
where $\mu$ is the known first moment of $X$. In \cite{bert02} it was shown 
that computing semiparametric bounds on \refer{eq:risk} using knowledge on the prices \refer{eq:price}
is a useful alternative to the standard methods of computing implied volatilities. For risk management purposes, semiparametric bounds can also be used to compute bounds on one sided deviations of the underlying risk, that is, its semivariance. Each of these common applications readily fit into the framework of the CG methodology presented in Section~\ref{sec:method}, demonstrating the variety of contexts in which the CG approach can be used to compute semiparametric bounds. 

\ssection{Summary}
The CG methodology presented here provides a practical optimization-based algorithm for computing semiparametric bounds on the expected payments of general insurance instruments and financial assets. Section~\ref{sec:method} described how the general problem described in \eqref{eq:ub} can be solved, in most practical instances, by solving a sequence of linear programs that are updated using simple arithmetic operations. The CG approach also readily provides a representation of the worst/best-case distributions associated with a semiparametric bound problem.

To illustrate the potential of the of the CG algorithm semiparametric bounds on the payoff of common insurance policies were computed.
It was also shown that additional distributional information such as continuity and unimodality can be incorporated in the formulation in a straightfoward fashion. The ability to include these constrains provides tighter bounds on the quantity of interest as well as  distributions that fit the practitioner's problem specific knowledge. 
\change{
Note that from the recent work of \citet{LeeL10}, it follows that for some Property/Casualty insurance problems it will be suitable to consider that the underlying random variable follows a distribution that is a mixture of Erlang distributions. The advantage of using mixtures of Erlang distributions is the existence of extremely efficient expectation--maximization (EM) algorithms for parameter estimation from raw data.
This interesting line of work will be the subject of future research work.}

The CG methodology offers a powerful and compelling alternative for computing semiparametric bounds 
in comparison with the main approaches used in the literature to compute them; namely, 
deriving analytical solutions to special cases of the problem or solving it numerically  using semidefinite programming.
This is due to the speed, generality, and ease of implementation of the CG algorithm.
The CG algorithm achieves accurate results at a very small computational cost. The straightforward implementation used for the test problems shown here generated solutions in at most 1-2 seconds. Furthermore, although the examples considered here focused on obtaining semiparametric bounds for insurance policies with piecewise linear payoff given moment information about the underlying loss, 
the CG approach presented here allows for a very general class of univariate semiparametric bounds to be computed using the same 
core the solution algorithm.

\section*{\centering Acknowledgements}
This work has been carried out thanks to funding from the Casualty Actuarial Society, Actuarial Foundation
of Canada, and Society of Actuaries.

\setcounter{equation}{0}
\numberwithin{equation}{section}

\section*{\centering Appendices}

\appendix 


\section{Smooth Approximations of the Uniform Distribution} \label{app:approx}
Let $X$ be a random variable that is uniformly distributed on the interval $[a,b]$. The probability density function (PDF) of $X$ is $f(x) = \frac{1}{b-a}$. It is possible to construct a smooth function that approximates $f(x)$ and is asymptotically equal to the true PDF. We define the following $\eta$ parameterized function.
\beq f_\eta(x) = \frac{1}{b-a} \left[\frac{1}{1+e^{-\eta(x-a)}} - \frac{1}{1+e^{-\eta(x-b)}}\right] \label{eq:uneta} \eeq
The probability function $f_\eta(x)$ is the difference in two shifted logistic functions. \begin{lemma}
\label{totpro}
For any $\eta > 0$, and $a,b \in \R$ such that $b \ge a$, the cumulative probability distribution $F_\eta(x)$ associated with the probability distribution $f_\eta(x)$ is $F_\eta(x) = 1 - \frac{1}{\eta(b-a)} \ln \left ( \frac{1 + e^{-\eta(x-b)}}{1 + e^{-\eta(x-a)}} \right )$. In particular, $\lim_{x \to \infty} F_\eta(x) = 1$, $\lim_{x \to -\infty} F_\eta(x) = 0$, and \change{$\frac{d F_\eta(x)}{d x} \ge 0$ for all $x \in \R$}.
\end{lemma}

In Lemma~\ref{l:u} we show that as $\eta \rightarrow \infty$ \eqref{eq:uneta} will converge to the PDF of $X$.

\begin{lemma}
	\label{l:u}
	As $\eta \rightarrow \infty$ the function $f_\eta(x)$ in \eqref{eq:uneta} converges to a uniform distribution on $[a,b]$.
\end{lemma}
\begin{proof}
	To show that $f_\eta(x) \rightarrow \Unif(a,b)$ as $\eta \rightarrow \infty$ consider three different cases corresponding to three intervals of $x$. First consider the case in which $x < a$. For $x < a$ each of the exponent terms are positive, i.e. $0 < -\eta(x-b) < -\eta(x-a)$, for all $\eta > 0$. So, for $x < a$ we see that $f_\eta(x) \rightarrow 0$. Next we look at $x > b$. In this case each of the exponent terms are negative, which again yields a limit of $0$. Finally, consider $a < x < b$. On this interval the exponent terms satisfy $-\eta(x-a) < 0 < -\eta(x-b)$ for all $\eta > 0$. So, for $a < x < b$ we have $f_\eta(x) \rightarrow \frac{1}{b-a} (1 - 0) = \frac{1}{b-a}$. We can conclude that the limit of $f_\eta(x)$ is as follows
\begin{align*}
	\underset{\eta \rightarrow \infty}{\lim} f_\eta(x) &= \left\{
		\begin{array}{l l}
			0 & x < a \\
			\frac{1}{b-a} & a < x < b \\
			0 & x > b
		\end{array} \right. \\
		&= \Unif(a,b)
\end{align*}
\end{proof}

To demonstrate Lemma~\ref{l:u}, consider $X \sim \Unif(20,30)$. The PDF of $X$ can be approximated using \eqref{eq:uneta} with $a = 20$ and $b = 30$. In Figure~\ref{fig:byeta} we plot the PDF of X as well as the approximation curve for different values of $\eta$.
\begin{figure}[H]
	\centering
	\includegraphics[width = 0.5\textwidth]{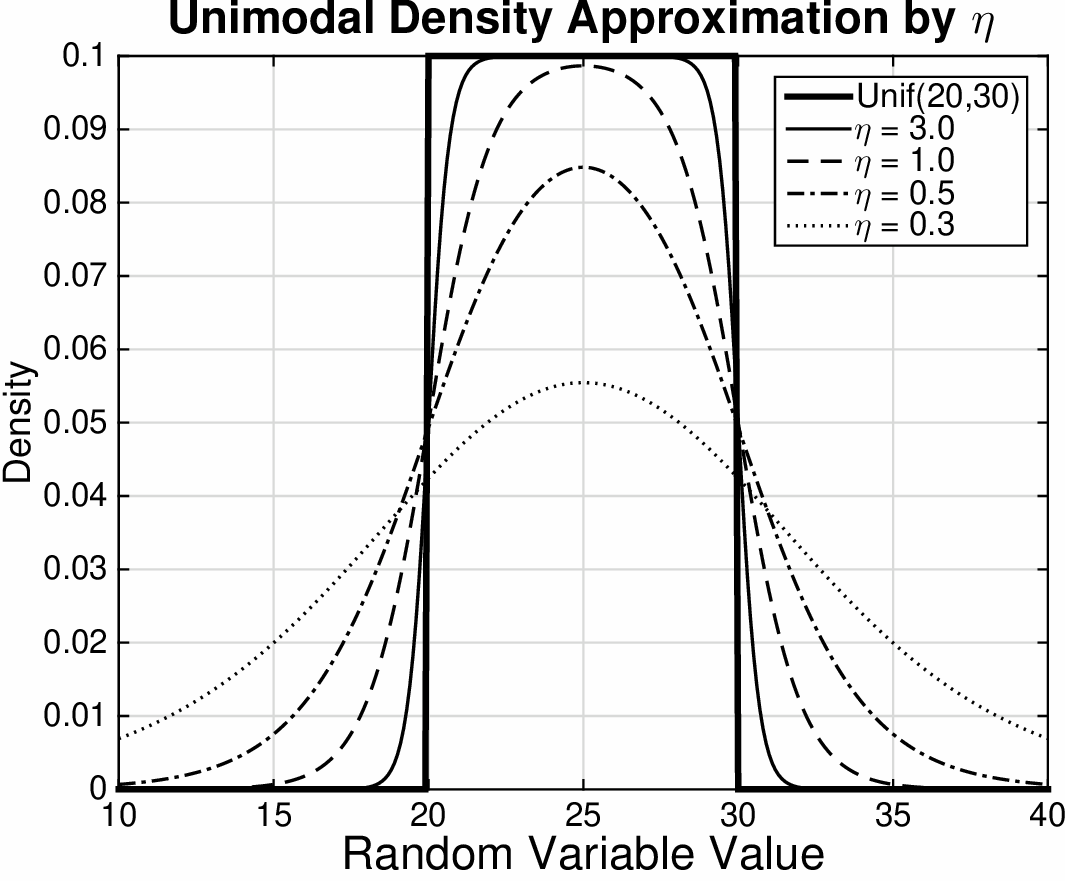}
	\caption{PDF of $\Unif(20,30)$ distribution along with the approximating density \eqref{eq:uneta} for different values of $\eta$.}
	\label{fig:byeta}
\end{figure}
From Figure~\ref{fig:byeta}, we see that as the parameter $\eta$ increases the curve of \eqref{eq:uneta} approaches the PDF of~$X$.


\end{document}